\documentclass[onecolumn,english]{sig-alternate}
\usepackage[T1]{fontenc}
\usepackage[latin9]{inputenc}
\usepackage{float}
\usepackage{amstext}
\usepackage{stackrel}
\usepackage{graphicx}

\makeatletter

\providecommand{\tabularnewline}{\\}
\floatstyle{ruled}
\newfloat{algorithm}{tbp}{loa}
\providecommand{\algorithmname}{Algorithm}
\floatname{algorithm}{\protect\algorithmname}

\makeatother

\usepackage{babel}
\begin{document}

\conferenceinfo{}{}

\title{Enhancing Reuse of Constraint Solutions to Improve Symbolic Execution\titlenote{this paper has been submitted to ISSTA 2015}}

\numberofauthors{3}

\author{\alignauthor Xiangyang Jia\\
\affaddr{State Key Lab of Software Engineering,Wuhan University}\\
\affaddr{Luojia hill Street, 229}\\
\affaddr{Wuhan, China}\\
\email{jxy@whu.edu.cn}\\
\alignauthor Carlo Ghezzi\\
\affaddr{Dipartimento di Elettronica e Informazione, Politecnico di Milano}\\
\affaddr{Via Golgi, 42}\\
\affaddr{Milano, Italy}\\
\email{carlo.ghezzi@polimi.it}\\
\alignauthor Shi Ying\\
\affaddr{State Key Lab of Software Engineering,Wuhan University}\\
\affaddr{Luojia hill Street 229}\\
\affaddr{Wuhan, China}\\
\email{yingshi@whu.edu.cn}\\
}

\date{20 Jan 2015}
\maketitle
\begin{abstract}
Constraint solution reuse is an effective approach to save the time
of constraint solving in symbolic execution. {\normalsize{}Most of
the existing reuse approaches are based on }syntactic{\normalsize{}
or semantic equivalence of constraints; e.g. }the Green framework
is able to reuse constraints which have different representations
but are {\normalsize{}semantically} equivalent, through canonizing
constraints into syntactically equivalent normal forms.{\normalsize{}
However, }syntactic/semantic{\normalsize{} equivalence is not a necessary
condition for reuse}\textemdash some constraints are {\normalsize{}not
}syntactically or {\normalsize{}semantically} equivalent, but their
solutions still have potential for reuse. Existing approaches are
unable to recognize and reuse such constraints. 

In this paper, we present GreenTrie, an extension to the Green framework,
which supports constraint reuse based on the logical implication relations
among constraints. GreenTrie provides a component, called L-Trie,
which stores constraints and solutions into tries, indexed by an implication
partial order graph of constraints. L-Trie is able to carry out logical
reduction and logical subset and superset querying for given constraints,
to check for reuse of previously solved constraints. We report the
results of an experimental assessment of GreenTrie against the original
Green framework, which shows that our extension achieves better reuse
of constraint solving result and saves significant symbolic execution
time.
\end{abstract}
%Full ACM categories cannot be yet represented in LyX without resorting to ELT

{Verification}

\keywords{}{constraint solving, symbolic execution, cache and reuse}

\section{Introduction\label{sec:introduction}}

Symbolic execution has been proposed as a program analysis technique
since the 1970's \cite{King1976}. It gained a lot of attention in
recent years as an effective technique for generating high-coverage
test cases and finding subtle errors in software applications \cite{Bounimova2013,Avgerinos2014}.
Symbolic execution works by exploring as many program paths as possible
in a given time budget, creating logical formulas encoding the explored
paths, using a constraint solver to check for feasible execution paths
and generate test cases, as well as finding corner-case bugs such
as buffer overflows, uncaught exceptions, and checking higher-level
program assertions \cite{Cadar2013,Pasareanu2009}.

In symbolic execution, constraint solving plays an important role
in path feasibility checking, test inputs generation, and assertions
checking. Since constraint satisfaction is a well-known NP-complete
problem, not surprisingly it is always the most time-consuming task
in symbolic execution. Despite significant advances in constraint
solving technology during the last few years\textemdash which made
symbolic execution appliable in practice\textemdash constraint solving
continues to be a bottleneck in symbolic execution \cite{Cadar2013,anand2012techniques}.
In order to ease constraint-solving in symbolic execution, some approaches
have been proposed, such as irrelevant constraint elimination \cite{Sen2005,Cadar2006},
incremental solving \cite{Cadar2008,Cadar2006}, and constraint solution
reuse \cite{Visser2012,Yang2012,Cadar2008}. 

The Green framework \cite{Visser2012} is a constraint solution reuse
framework which stores the solutions of constraints and reuses them
across runs of the same or different programs. Green stores constraints
and their solutions as key-value pairs in an in-memory database Redis
\cite{zawodny2009redis}, and queries the solutions for reuse based
on string matching. To improve the matching ratio, all the constraints
are sliced and canonized before they are stored and queried. \textit{Slicing}
is a process to obtain the minimal constraint required for satisfiability
checking. Because path conditions in symbolic execution are always
generated by conjoining a new term to an old satisfiable constraint,
the slicing process removes the old constraints which are irrelevant
to the new path condition, based on graph reachability checking. Slicing
has the potential to significantly reduce both the number of constraints
and the number of variables in the problem. \textit{Canonization}
represents each individual constraint into a normal form. Linear integer
sub-constraints are converted into a normal form $ax+by+cz+...+k\, op\,0$,
where $op\in\{=,\neq,\leq\}$. In addition, canonization sorts the
constraint in a lexicographic order, and renames the variables into
a standard form. For example, after canonization, the constraint $x+y<z\bigwedge x=z\bigwedge x+10>y$
becomes $-v0+v1-9\le0\bigwedge v0+v1-v2+1\le0\bigwedge v0-v2=0$.
As a consequence of slicing and canonization, a constraint may become
syntactically equivalent to a previously evaluated constraint and
thus simple string matching may detect a potential reuse.

Most reuse approaches for constraint solution are based on syntactic
or semantic equivalence, e.g., \cite{Makhdoom2014,Hossain2014,Yang2012,Chen2014}.
However, syntactic/semantic equivalence is not a necessary condition
for reuse\textemdash some constraints are not equivalent, but there
is still potential for reuse. Here are some examples: 
\begin{itemize}
\item Example 1: Suppose we have proved constraint $x>0\bigwedge x<y\bigwedge y-x>1$
to be satisfiable, with a solution\{x:1, y:3\}. Constraint $x<y\bigwedge y-x>1$
can also be proved to be satisfiable by reusing this solution.
\item Example 2: Suppose we have proved constraint $x<0\bigwedge x>1$ to
be unsatisfiable. Constraint $x<0\bigwedge x>1\bigwedge x\neq10$
can also be proved to be unsatisfiable by reusing this result.
\item Example 3: Suppose we have proved constraint $x<-1$ to be satisfiable
with a solution \{x:-5\}. Constraint $x<0\wedge x\neq-1$ can also
be proved to be satisfiable by reusing this solution.
\item Example 4: Suppose we have proved constraint $x<0\bigwedge x>1$ to
be unsatisfiable. Constraint $x<-1\bigwedge x>2$ can also be proved
to be unsatisfiable by reusing this result.
\end{itemize}
As far as we know, no existing approach (including Green) can reuse
constraint solutions in all of above situations. KLEE \cite{Cadar2008}
is able to cope with Examples 1 and 2, through subset and superset
matching, but is unable to cope with Examples 3 and 4. 

In this paper, we present GreenTrie, an extension to the Green framework,
which supports constraint reuse based on the logical implication relations
among constraints. GreenTrie provides a component, called L-Trie,
which stores constraints and solutions into tries (an ordered tree
data structure that is used to store a dynamic set or associative
array \cite{Bodon2003}), indexed by an implication partial order
graph of constraints. L-Trie is able to carry out logical reduction
and logical subset and superset querying for given constraints, to
check for reuse of previously solved constraints. This approach supports
constraints reuse based on their logical implication relations.

The contributions of this paper can be summarized as follows:
\begin{itemize}
\item We present a theoretical basis for checking constraint reusability
based on their logical relationship, and give rules to check the implication
relationship between linear integer arithmetic constraints.
\item We present a constraint reduction approach to reduce the constraint
into more concise form, as well as to find obviously conflicting sub-constraints.
\item We describe the L-Trie data structure, which is used to cache past
constraint solutions into tries indexed by implication partial order
graphs.
\item We give logical superset and subset checking algorithms to check the
existence of reusable solutions stored in L-Trie.
\item We evaluate the performance of GreenTrie in three scenarios: (1) reuse
in a single run of the program, (2) reuse across runs of the same
program, (3) reuse across different programs. The experiments show
that, compared to original Green framework using the Redis store,
GreenTrie achieves better reuse of constraint solving results, and
saves significant time in symbolic execution.
\end{itemize}

\section{Logical Basis of our Approach\label{sec:logical-basis}}

Constraint satisfiability checking\textemdash the quintessential NP-complete
problem\textemdash has been studied extensively, with strong motivations
arising especially from artificial intelligence. A (finite domain)
constraint satisfaction problem can be expressed in the following
form: given a set of variables, together with a finite set of possible
values that can be assigned to each variable, and a list of constraints,
find values of the variables that satisfy every constraint\cite{brailsford1999constraint}. 

In symbolic execution scenarios, the target of constraint solving
is to find a solution for given constraint (always in the form of
a conjunction of several sub-constraints). The \textit{solution},
if it exists, is a valuation function mapping the set of variables
of a constraint to a value set. If we substitute the variables in
the constraint with the values in the solution, the constraint evaluates
to TRUE. When a solution exists, the constraint is \textit{satisfiable};
if not, it is \textit{unsatisfiable}. In this paper we focus on \textit{linear
integer constraints}, for which satisfiability is decidable. In our
future work we plan to extend our approach to also cope with other
kinds of constraints, such as non-linear constraints and string constraints. 

\newtheorem{lemma}{Lemma}
\begin{lemma}\label{thm:thm1}Given two constraints $C$ and $C'$, (1) if C is
satisfiable and has a solution V, and $C\rightarrow C'$, then $C'$is
satisfiable and V is also a solution of C'. (2) if $C$ is unsatisfiable
and $C'\rightarrow C$, then $C'$ is unsatisfiable.\end{lemma}
\begin{proof}
(1) Because C is satisfiable and has a solution V, by substituting
the variables in the constraint with the values in solution V, C evaluates
to TRUE. Since $C\rightarrow C'$, according to the definition of
logical implication, $C'$ evaluates to TRUE for this substitution
too. Therefore, V is also a solution for $C'$ and $C'$ is satisfiable.
(2) If C is unsatisfiable, $\neg C$ will evaluate to TRUE for all
valuations. Since $C'\rightarrow C$, then $\neg C\rightarrow\neg C'$
and hence $C'$ will evaluate to FALSE for all valuations.i.e $C'$
is unsatisfiable.
\end{proof}
According to Lemma \ref{thm:thm1}, checking the implication relationship
between constraints can be a basis for reusing constraint satisfiablity
checks. In symbolic execution, constraints are mainly utilized to
represent the path conditions of branches in code, and each of them
is a conjunction of all the branching conditions (in terms of the
program inputs) form the first branch to current location. Therefore,
a constraint is always in the form $C_{1}\land C_{2}...\land C_{n}$,
and has a sub-constraint set $\{C_{1},C_{2}...C_{n}\}$. In our approach,
we will check the reusability of such constraints through querying
\textit{logical subsets} and\textit{ logical supersets} of the sub-constraint
set in the solution store.

\newdef{definition}{Definition}
\begin{definition}\label{def:def1}(Logical subset and logical superset) Given two constraint
sets X and Y, if $\forall_{x\in X}\exists_{y\in Y}y\rightarrow x$,
then X is a logical subset\textbf{ }of Y and Y is a logical superset
of X.

\end{definition}

For example, if X = \{x$\neq$0, x>-1, x<2\}, Y=\{x>1, x<2\}, because
$x>1\rightarrow x\neq0$, $x>1\rightarrow x>-1,\, x<2\rightarrow x<2$,
then X is a logical subset of Y, and Y is a logical superset of X,
even though Y has less elements than X.

\newtheorem{theorem2}{Theorem}
\begin{theorem2}\label{thm:thm2}Given two constraints in conjunctive form $C=\stackrel[i=1]{n}{\bigwedge}C_{i}$,
$C'=\stackrel[i=1]{m}{\bigwedge}C'_{i}$, where C has a sub-constraint
set $S=\{C{}_{1},C{}_{2}...C{}_{n}\}$, and $C'$ has a sub-constraint
set $S'=\{C'{}_{1},C'{}_{2}...C'{}_{m}\}$, (1) if C is satisfiable
and has a solution V, and S is a logical superset of $S'$, then $C'$is
satisfiable and V is also a solution of C'. (2) if C is unsatisfiable,
and $S$ is a logical subset of $S'$, then $C'$is unsatisfiable.
\end{theorem2}
\begin{proof}
(1) Since $S$ is a logical superset of $S'$, $\forall_{c'\in S'}\exists_{c\in S}$
$c\rightarrow c'$. Hence $C_{1}\land C_{2}...\land C_{n}\rightarrow C'_{1}\land C'_{2}...\land C'_{m}$,
i.e. $C\rightarrow C$'. According to Lemma \ref{thm:thm1}, if C
is satisfiable and has a solution V, then $C'$ is satisfiable and
V is also a solutions for C'. (2) Since $S$ is a logical subset of
$S'$, $\forall_{c\in S}\exists_{c'\in S'}c'\rightarrow c$. Hence
$C'_{1}\land C'_{2}...\land C'_{m}\rightarrow C_{1}\land C_{2}...\land C_{n}$,
i.e.$C'\rightarrow C$. According to Lemma \ref{thm:thm1}, if $C$
is unsatisfiable, then $C'$ is unsatisfiable. 
\end{proof}
According to Theorem \ref{thm:thm2}, a constraint can be shown to
be satisfiable if a logical superset can be retrieved in a storage
that caches satisfiable sub-constraint sets. Likewise, a constraint
can be shown to be unsatisfiable if a logical subset can be retrieved
in a storage that caches unsatisfiable sub-constraint sets. 

\textbf{Normal form of linear integer constraint. }In this paper,
every atomic linear integer constraint is canonized into the form:
\[
h_{1}v_{1}+h_{2}v_{2}+h_{3}v_{3}+...h_{n}v_{n}+k\, op\,0\quad
\]
where $v_{1},v_{2}...v_{n}\,$are distinct variables, the coefficients
$h_{1}$, $h_{2}...$, $h_{n}$ are numeric constants, k is an integer
constant, $\, h_{1}\geq0$, and $\, op\in\{=,\neq,\leq,\geq\}$. The
expression $h_{1}v_{1}+h_{2}v_{2}+h_{3}v_{3}+...h_{n}v_{n}$, which
contains all non-constant terms, is the constraint's \textit{non-constant
prefix. }

\textbf{Implication Checking Rules}. We define a list of rules to
check for specific implication relationships between two atomic linear
integer constraints. In this paper, only constraints which have the
same non-constant prefix can be checked by rules. In the future, we
plan to extend the rules to handle more complex situations.We compare
non-constant prefixes based on string comparison and constant values
based on numeric comparison, which is quite efficient. The implication
checking rules are listed below. In these rules, $P$ is a non-constant
prefix and $n$ is a constant value. The rules enable checking the
implication relationship between linear integer arithmetic constraints
with operators $=,\neq,\leq,\geq$. 

\[
(R1)\:\frac{}{C\rightarrow C}\qquad\qquad\qquad\quad\quad(R2)\:\frac{n\neq n'}{P+n=0\rightarrow P+n'\neq0}
\]

\[
(R3)\:\frac{n\geq n'}{P+n=0\rightarrow P+n'\leq0}\quad(R4)\:\frac{n\leq n'}{P+n=0\rightarrow P+n'\geq0}
\]

\[
(R5)\:\frac{n\text{>}n'}{P+n\leq0\rightarrow P+n'\neq0}\quad(R6)\:\frac{n>n'}{P+n\leq0\rightarrow P+n'\leq0}
\]

\[
(R7)\:\frac{n\text{<}n'}{P+n\geq0\rightarrow P+n'\neq0}\quad(R8)\:\frac{n<n'}{P+n\geq0\rightarrow P+n'\geq0}
\]

\section{Overview of GreenTrie}

GreenTrie extends the Green framework to improve the reuse of constraint
solutions. The overview architecture of GreenTrie is illustrated in
Fig.\ref{fig:overview}. GreenTrie includes a component named L-Trie,
which replaces the Redis store of the original Green framework. L-Trie
is a bipartite store used for caching satisfiable and unsatisfiable
constraints, respectively, each composed of a constraint trie and
its logical index. The constraint trie stores constraints in the form
of sub-constraint sets, and the logical index is a partial order graph
of implication relations for all the sub-constraints in the trie. 

L-Trie and Green work together within GreenTrie. Any request to solve
a constraint is handled by Green through the following four steps:
\textit{(1) slicing:} it removes pre-solved irrelevant sub-constraints;
\textit{(2) canonization:} it converts a constraint into normal form;
\textit{(3) reusing:} it queries the solution store for reuse; if
a reusable result is not retrieved, (4) \textit{translation:} the
constraint is translated into the input format required by the chosen
constraint solver (such as CVC3\cite{Barrett2007}, Z3, Yices\cite{Dutertre2006},
or Choco), which is then invoked to solve the constraint from scratch.
The result produced by the constraint solver is finally stored into
either satisfiable constraint store(SCS) or unsatisfiable constraint
store(UCS)(Fig.\ref{fig:overview}). 

L-Trie provides three interfaces to Green: constraint reduction, constraint
querying, and constraint storing. These are presented in detail in
the following sections. \textit{Constraint reduction} is performed
after the constraint is canonized by the Green framework; redundant
sub-constraints are removed and conflicting sub-constraints are reported
in this phase. \textit{Constraint querying} handles the requests issued
by Green to retrieve pre-solved constraints. Based on Theorem \ref{thm:thm2},
it checks whether the constraint has a logical superset in the satisfiable
constraint store or has a logical subset in the unsatisfiable constraint
store. \textit{Constraint storing} splits solved constraint into sub-constraints,
puts them into the corresponding constraint trie, and the also updates
the logical index.

\begin{figure*}
\begin{centering}
\includegraphics[clip,scale=0.8]{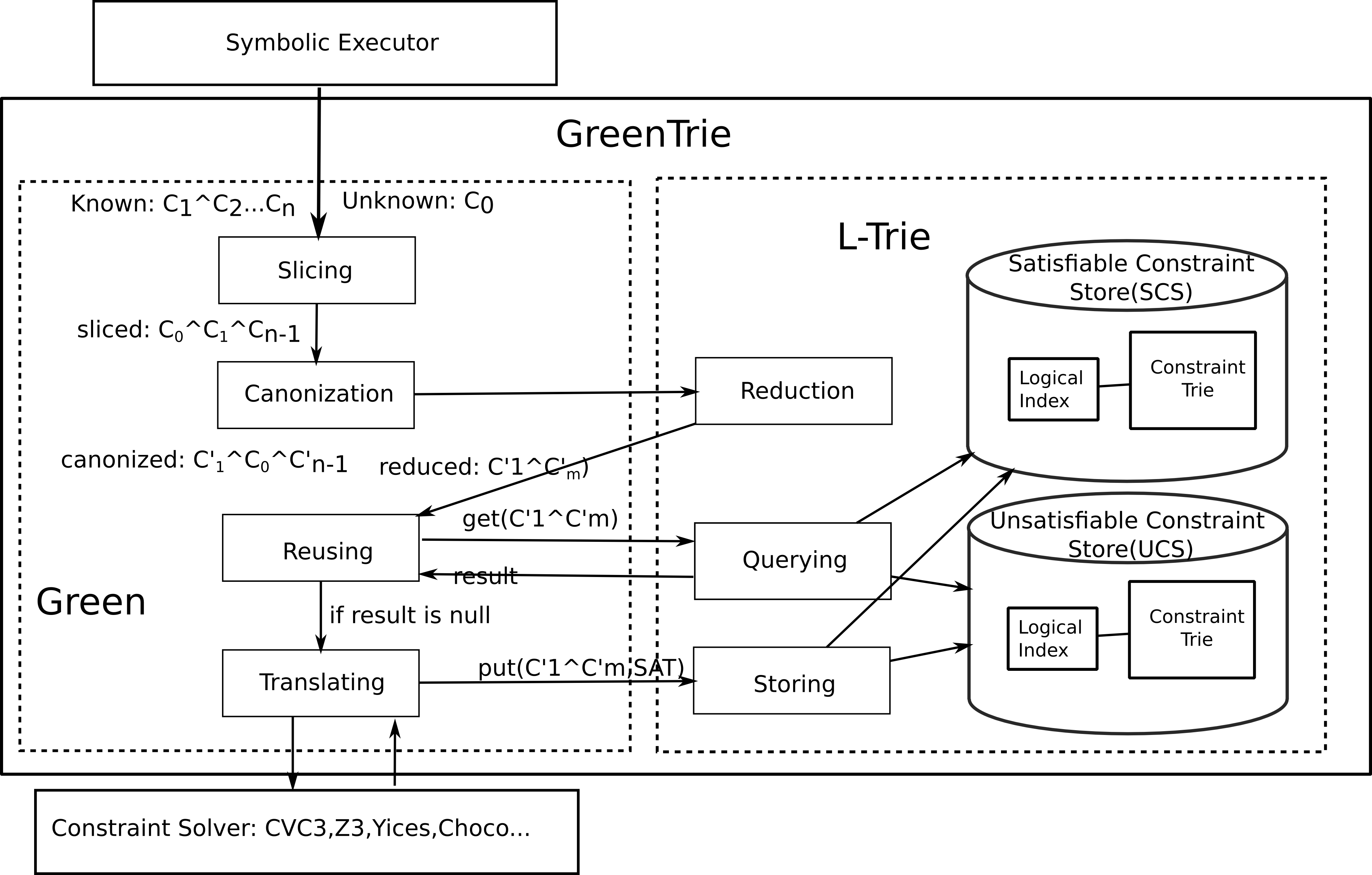}
\par\end{centering}

\protect\caption{\label{fig:overview}The overview architecture of GreenTrie}
\end{figure*}

\section{Constraint Reduction }

Symbolic execution conjoins constraints as control flow branches are
traversed. This may introduce redundant sub-constraints, where a sub-constraint
is implied by another. For example, if constraint\textit{ x$\geq$0}
is conjoined to constraint $x\neq$-2, the latter becomes redundant
and can be eliminated. It may also happen that one can easily detect
that the newly added constraint conflicts with another constraints,
making the whole constraint unsatisfiable; for example, consider the
case where \textit{x=0} is conjoined with\textit{ x$\geq$3}. Constraint
reduction in our approach is able to recognize such situations: it
can both reduce the constraint into more concise form and also find
obviously-conflicted sub-constraints. As we mentioned, we only focus
on the linear integer arithmetic constraints. In the future, we plan
to reduce other kind of constraints based on term rewriting \cite{Braione2013}. 

Our approach performs reduction as follows. The sub-constraints with
same non-constant prefix are merged and reduced based on their \textit{value
interval of non-constant prefixes}. For example, considering constraint
x+y+3$\leq$0, its non-constant prefix x+y has a value interval $[MIN,-3]$,
and for constraint x+y$\geq$0, the value interval is $[0,MAX]$.
As for constraint x+y+4=0, the value interval is $[4,4]$. If the
constraint is stated as an inequality, as for example x+y+6$\neq$0,
we have two value intervals $[MIN,-6)$ and $(-6,MAX]$. Equivalently,
we can represent this situation by introducing the concept of an \textit{exceptional
point} (in this case, \textquotedbl{}-6\textquotedbl{} ).

To support reduction, firstly all sub-constraints with the same non-constant
prefix are merged together, by computing the overlapping interval
$[A,B]$ of these constraints, and at the same time collecting the
exceptional points into a set $E$. For example, after computing of
constraint $x+y+3\geq0\wedge x+y+5\geq0\wedge x+y-4\leq0\wedge x+y\neq0\wedge x+y+6\neq0\wedge x+y-4\neq0$,
we get an overlapping interval $[-3,4]$ and an exceptional point
set $E=\{-6,0,4\}$. After this, we go through the following steps:
\begin{enumerate}
\item We discard all exceptional points that are outside the overlapping
interval; in the example, the value of $E$ becomes $\{0,4\}$.
\item If one endpoint of the overlapping interval A (or B) belongs to $E$,
we (repeatedly) change its value and eliminate A (or B) from E at
the same time. In the example after this step the interval becomes
$[-3,3]$ and the new value of $E$ is $\{0\}$.
\item If the overlapping interval is empty then the constraint is unsatisfiable
and we report a conflict; otherwise we translate $[A,B]$ and $E$
into a constraint in normal form. In the example, the final result
of our reduction is $x+y+3\geq0\wedge x+y-3\leq0\wedge x+y\neq0$.
\end{enumerate}

\section{Constraint Storing}

L-Trie provides a different storage scheme that replaces the Redis
store of Green: 
\begin{itemize}
\item Unlike Redis, which stores the strings representing constraints and
solutions as key-value pairs, L-Trie splits constraints into sub-constraint
sets, and stores them into tries, in order to support logical subset
and superset queries based on Theorem \ref{thm:thm2}.
\item L-Trie stores unsatisfiable and satisfiable constraints into separate
areas: the \textit{Unsatisfiable Constraint Store} (UCS) and the \textit{Satisfiable
Constraint Store} (SCS) respectively. The two areas are organized
differently to efficiently support logical subset querying and logical
superset querying, which pose different requirements.
\item L-Trie maintains a logical index for each of the two tries, to support
efficient check of the implication relations. The logical index is
represented as an implication partial order graph (IPOG), whose nodes
contain references to nodes in the trie.
\end{itemize}
\begin{figure*}
\begin{centering}
\includegraphics[scale=0.9]{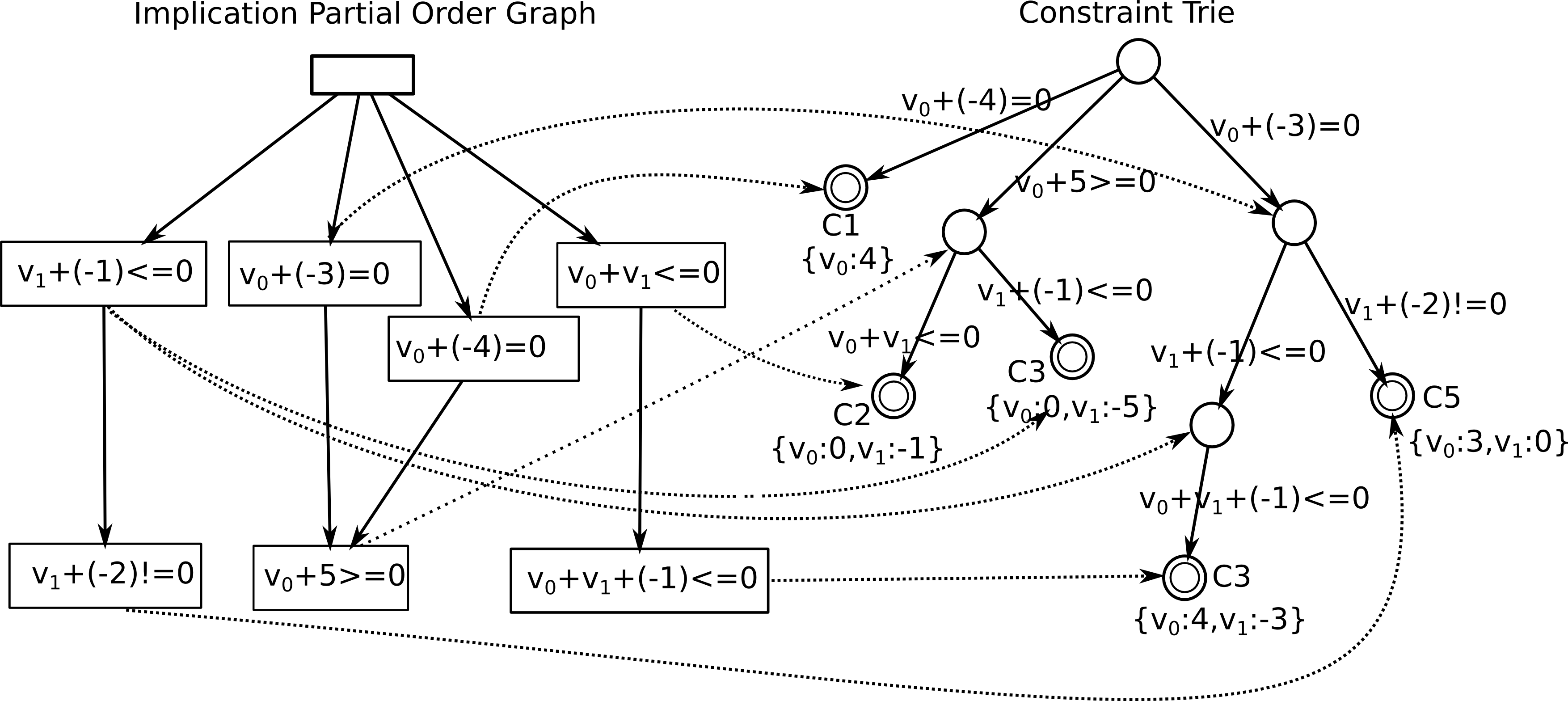}
\par\end{centering}

\protect\caption{\label{fig:trie-store}The structure of constraint stores in L-Trie
(both UCS and SCS have the same structure).}
\end{figure*}

Both UCS and SCS have the same structure (see Fig. \ref{fig:trie-store}).

\textbf{Constraint Trie}. The constraint trie is designed to store
a sub-constraint set of solved constraints. The sub-constraint set
is sorted in lexicographic order based on string comparison, to guarantee
that sub-constraints with same non-constant prefix are kept close
to each other. The labels of the constraint trie record the sub-constraints.
The leaf nodes indicate the end of the constraint and are annotated
with the solution (the solution is null for the leaves of the UCS
trie). As shown in Fig.\ref{fig:trie-store}, the leaf node C2 corresponds
to a constraint $v_{0}$+5>=0 $\wedge$$v_{0}$+$v_{1}$<=0, which
has a solution \{$v_{0}:0$,$v_{1}:-1$\}, and its sub-constraints
$v_{0}$+5>=0 and $v_{0}$+$v_{1}$<=0, are annotated as edge labels
in the path. 

If a constraint \textit{C }is a conjunction of atomic constraints
that is a prefix of another constraint \textit{C'}(e.g. C is $A\wedge B$,
and \textit{C' is }$A\wedge B\wedge C$,), only one of them is kept
in the trie. We keep the longer constraint in the SCS trie, while
we keep the shorter in the UCS trie. 

\textbf{Implication Partial Order Graph (IPOG)}. IPOG is a graph that
contains all the atomic sub-constraints appearing in its associated
constraint trie, and arranges them as a graph based on the partial
order defined by the implication relation. With this graph, given
a constraint C, we can query the sub-constraints which imply C, as
well as the sub-constraints which C implies, as we will see later.
This is useful to improve the efficiency of implication checking in
logical subset and superset querying. IPOG nodes are labeled by a
sub-constraint and have references to all trie nodes whose input edge
is labeled with exactly this sub-constraint. Through these references,
it is possible to trace all the occurrences of a given sub-constraint.

\textbf{Storing the constraints. }Everytime a constraint is solved
(or it is proved to be unsatisfiable), SCS (respectively, UCS) must
be updated to store possibly new sub-constraints that were not found
before, as we describe hereafter. Let $C=C_{1}\land C_{2}...\land C_{n}$
be the solved constraint in canonical form. Constraint $C$ can be
represented by the set <$C_{1},C_{2},...C_{n}$>, where each element
is an atomic sub-constraints. This set is sorted by the lexicographic
order that yields the canonical form. $C_{1}$ (respectively, $C_{n}$)
is called the leftmost (respectively, rightmost) sub-constraint of
C. The storage procedure proceeds as follows:
\begin{enumerate}
\item Starting from the trie root node, we consider the (possibly empty)
maximal path whose labels coincide with a prefix $C_{1}C_{2}...C_{i}$
of C %
\footnote{Note that this procedure ensures that the UCS trie stores the shortest
of any two unsatisfiable constraints where one is a prefix of the
other, while the SCS trie stores the longest.%
}.

\begin{enumerate}
\item If $C_{i}$ labels the input edge of a leaf node, it means that we
found a logical subset of C in the trie. In the case of the SCS trie,
we remove the solution labeling the leaf and append to the leaf a
linear subtree with edges labeled $C_{i+1}...C_{n}$. In the case
of the UCS trie, we simply ignore constraint C, which is not saved. 
\item If i=n and we have not reached a leaf node, it means that we found
a logical superset of C in the trie. In the case of the SCS trie,
we ignore constraint C and we do not save it. In the case of the UCS
trie, we delete the subtree rooted $C_{i}$ and the node labeled $C_{i}$
becomes a leaf, which is labeled with C's solution. 
\item Otherwise, we append a linear subtree with edges labeled $C_{i+1}...C_{n}$
$C_{i+1}...C_{n}$ to the trie node labeled $C_{i}$and add C's solution
to the last node labeled $C_{n}.$
\end{enumerate}
\item During step 1, whenever we add a new sub-constraint, it will also
be stored into IPOG in a way that preserves the partial order defined
by the implication relation among atomic sub-constraints.
\end{enumerate}

\section{Constraint Querying }

According to Theorem \ref{thm:thm2}, if we want to find a solution
for a constraint which has the constraint set C, we should check if
any logical subset of C exists in UCS, or if any logical superset
of C exists in SCS. Since the constraints are stored in tries, checking
for logical subset means that we should find a path from root to a
leaf node in the UCS trie so that each constraint in the path is implied
by one of the constraints in C. And checking for logical superset
means that we should find a path in the SCS trie, so that each constraint
in C is implied by one of the constraints in the path.

\subsection{Implication Set and Reverse Implication Set}

To support efficient check of implication between constraints in C
and constraints in trie paths, we introduce the notions of \textit{implication
set (IS) }and \textit{reverse implication set} \textit{(RIS)} of an
atomic constraint $\varphi$: \textit{IS}($\varphi$) contains all
the atomic constraints in UCS which $\varphi$ implies, whereas \textit{RIS}($\varphi$)
contains all the constraints in SCS which imply $\varphi$. With the
help of \textit{IS} and \textit{RIS}, implication checking can be
reduced to checking the existence of constraints in sets.

\textit{IS}($\varphi$) is built by searching the UCS IPOG to find
all the constraints in IPOG which $\varphi$ implies, and \textit{RIS}($\varphi$)
is built by searching the SCS IPOG to find all the constraints in
IPOG which imply $\varphi$. Instead of visiting the whole IPOGs,
we only visit the sub-graph which has the same non-constant prefix
as $\varphi$, since (see Section 2) we exploit the implication relationship
between two atomic constraints when they have same non-constant prefix.
Because such sub-graphs are often small, the task of building these
two sets is always very fast.

\subsection{Logical Superset Checking Algorithm}

\begin{algorithm}
\protect\caption{Logical superset checking algorithm}

\textit{/{*} Check if logical superset of C exists in SCS trie; C
is the constraint set to be checked. In this function, rmostRISof(L)
is the last element of L, i.e. it is the RIS of the rightmost atomic
sub-constraint of C; nodesInTrie(c) is the set of trie nodes referenced
by c in IPOG{*}/}

\textbf{Function} boolean checkSuperset(C, IPOG, Trie)
\begin{enumerate}
\item L := empty list; \textit{//the list of RIS}
\item \textbf{for each} atomic sub-constraint c in C \textbf{do}
\item ~~~~$S$ := RIS(c, IPOG);
\item ~~~~\textbf{if }$S=\textrm{Ø}$\textbf{ then} \textbf{return}
false \textbf{else} L.add(S); 
\item \textbf{for each} c in rmostRISof(L)\textbf{ do }
\item ~~~~\textbf{for each} n in nodesInTrie(c) \textbf{do }
\item ~~~~~~~~\textbf{if }isSuperset(node, L) \textbf{then} \textbf{return}
true;
\item \textbf{return} false;
\end{enumerate}
\textit{/{*} Check if the constraints on the path is a logical superset
of the constraint; n is the start node of path; L is the list of RIS
{*}/}

\textbf{Function} boolean isSuperset(n,L) 
\begin{enumerate}
\item cur:=n;//current node
\item pos:=s.size-1; \textit{//current position of L}
\item \textbf{while} $cur\neq root$ \textbf{do}
\item ~~~~\textbf{while} $cur.in\in L[pos]$ \textbf{do} 
\item ~~~~~~~~pos:=pos-1;
\item ~~~~~~~~\textbf{if} $pos<0$ \textbf{then} \textbf{return}
true;
\item ~~~~cur:=cur.previous; 
\item \textbf{return} false;\end{enumerate}
\end{algorithm}

We present an algorithm to check the logical superset of constraint
set C in SCS. This algorithm (Algorithm 1) visits the trie bottom-up,
from the nodes whose input edges are labeled with constraints that
imply the rightmost atomic sub-constraint of C, moving up towards
the root node, and checking if the constraints on the path imply the
constraints in C.

Function checkSuperset has three parameters: C is a (lexicographically)
sorted constraint set to be queried, IPOG and Trie are the implication
partial order graph and the constraint trie in SCS. As shown in lines
1--4 of function checkSuperset, we first build the RIS for each constraints
in C and put them into a list L. If one constraint's RIS is empty,
then the function returns false, indicating that a logical superset
cannot be found in SCS. Lines 5--7 check all the trie nodes referenced
by the elements contained in the last RIS of list L; i.e., the nodes
whose input edge's labeling constraints imply the rightmost sub-constraint
of C. For each of these nodes, function isSuperset checks whether
the constraint set on the path from the node to the root is a logical
superset of C. If we find such path, then the function returns true,
otherwise it returns false. Function isSuperset has two parameters:
n is the start node and L is a list of RIS corresponding to each sub-constraint
of C. Lines 3--7 visit the trie path from the start node upward to
the root. Lines 4--6 repeatedly check if the constraint labeling the
incoming edge to the current node is an element of RIS. We use a loop
instead of a branch, because it is possible that one constraint on
the path implies several constraints in C. Line 6 indicates that if
every constraint in C is implied by constraints on the path, then
a logical superset is found.

Algorithm 1 shows the benefit on performance of using IPOG as a logic
index. Instead of visiting all the trie paths, it only visits a small
set of paths from the nodes whose input constraints imply the rightmost
sub-constraint of C.

\subsection{Logical Subset Checking Algorithm}

\begin{algorithm}
\protect\caption{Logical subset checking algorithm}

\textit{/{*} Check if logical subset of C exists in UCS trie; C is
the constraint set to be checked {*}/}

\textbf{Function} boolean checkSubset(C, IPOG,Trie) 
\begin{enumerate}
\item S := \{\}; \textit{//S represents the union of ISs}
\item \textbf{for each} atomic c in C \textbf{do}
\item ~~~~S := S $\cup$ IS(c, IPOG);
\item \textbf{if }$S\neq\textrm{Ø}$\textbf{ then} \textbf{return} hasSubset
(Trie.root, S) 
\item ~~~~\textbf{else} \textbf{return} false;
\end{enumerate}
\textit{/{*}Recursively check if any logical subset exists in the
sub-tree; n is the root of sub-tree ; S is a union set of ISs.{*}/}

\textbf{Function} boolean hasSubset(n,S) 
\begin{enumerate}
\item \textbf{if} n is leaf \textbf{then} \textbf{return} true; 
\item \textbf{for each} edge in n.out \textbf{do}
\item ~~~~\textbf{if} $edge.label\in S$ \textbf{then }
\item ~~~~~~\textbf{~if} hasSubset (n.next(edge), S) \textbf{then}
\textbf{return} true;
\item \textbf{return} false;\end{enumerate}
\end{algorithm}

This section presents an algorithm (Algorithm 2) to check for a logical
subset of constraint set C in UCS. The algorithm visits the trie top-down,
starting from the root, and selects successive nodes whose input constraints
are implied by constraints in C, until a leaf node is reached.

In Algorithm 2, function checkSubset has three parameters: a (lexicographically)
sorted constraint set C, and the UCS IPOG and Trie. Lines 2--3 build
the union set S of all ISs of atomic constraints in C. If S is not
empty (Lines 4--5), function hasSubset is invoked to check if a path
exists whose constraints are the logical subset of C. If S is empty,
then the function returns false, indicating that no  logical subset
can be found in the trie. Function hasSubset is implemented as a recursive
visit of the trie.

By building the union of all ISs of sub-constraints, this algorithm
significantly decreases the complexity of implication checking among
edge labels and sub-constraints in C and improves the performance
 of logical subset checking.

\section{Evaluation\label{sec:Evaluation}}

This section presents an experimental evaluation of the performance
of the GreenTrie framework. We compare the performance of GreenTrie
with the original Green framework which uses the Redis store and also
with a situation where no reuse is made of constraint solutions. The
assessment is performed by considering three scenarios: (1) reuse
in a single run of the program, (2) reuse across runs of the same
program, (3) reuse across different programs. 

All experiments were conducted on a PC with a 2.5GHz Intel processor
with 4 cores and 4Gb of memory. It runs the Centos 7.0 operating system.
We implemented the GreenTrie framework, and integrated it into the
well-known symbolic executor Symbolic Pathfinder\cite{Pasareanu2013,Anand2007}.

The experiments that follow are based on six programs which were used
in \cite{Visser2012}: 
\begin{itemize}
\item TriTyp implements DeMillo and Offutt\textquoteright s solution of
Myers\textquoteright s triangle classification problem ; 
\item Euclid implements Euclid\textquoteright s algorithm for the greatest
common divisor using only addition and subtraction; 
\item TCAS is a Java version of the classic anti-Collision avoidance system
available from the SIR repository; 
\item BinomialHeap is a Java implementation of binomial heap; 
\item BinTree implements a binary search tree with element insertion, deletion;
\item TreeMap uses a red-black tree to implement a Java Map-like interface.
\end{itemize}
In all the tables that summarize our experimental results we use the
following conventions:
\begin{itemize}
\item $t_{0}$, and $n_{0}$ denote the running time and the number of SAT
solving invocations, respectively, for classical symbolic execution
without any reuse;
\item $t_{1}$, and $n_{1}$ denote the running time and the number of SAT
solving invocations, respectively, when Green is used;
\item $t_{2}$, and $n_{2}$ denote the running time and the number of SAT
solving invocations, respectively, when GreenTrie is used;
\item \textit{$T$} = $(t_{0}-t_{2})/t_{0}$ denotes the \textit{time saving
ratio};
\item \textit{$R=(n_{0}-n_{2}\text{)}/n_{0}$} denotes the \textit{reuse
ratio};
\item \textit{$T'$} =$(t_{1}-t_{2})/t_{1}$ denotes the \textit{time improvement
ratio}; 
\item \textit{$R'=(n_{1}-n_{2}\text{)}/n_{1}$} denotes the \textit{reuse
improvement ratio}. 
\end{itemize}

\subsection{Reuse in a Single Run\label{sub:Reuse-in-Single}}

The first experiment evaluates performance of GreenTrie in a scenario
of reuse within a single run. To evaluate how performance scales with
the size of a symbolic execution tree, we modify the loop bound of
the TreeMap, BinTree, and BinomialHeap programs, thus yielding three
versions for each of these three programs.\textit{ }The results are
shown in Table \ref{tab:tab1}.

\begin{table*}
\protect\caption{\label{tab:tab1}Experiment result of reuse in single run }

\centering{}%
\begin{tabular}{lrrrrrrrrr@{\extracolsep{0pt}.}lr@{\extracolsep{0pt}.}l}
\hline 
Program & $n_{0}$ & $n_{1}$ & $n_{2}$ & $R$ & $R'$ & $t_{0}$(ms) & $t_{1}$(ms) & $t_{2}$(ms) & \multicolumn{2}{c}{$T$} & \multicolumn{2}{c}{$T'$}\tabularnewline
\hline 
Trityp & 32 & 28 & 28 & 12.50\% & 0.00\% & 1040 & 1005 & 1073 & -3&17\% & -6&77\%\tabularnewline
Euclid & 278 & 249 & 222 & 20.14\% & 10.84\% & 5105 & 5996 & 5884 & -15&26\% & 1&87\%\tabularnewline
TCAS & 680 & 41 & 14 & 97.94\% & 65.85\% & 12742 & 3356 & 2275 & 82&15\% & 32&21\%\tabularnewline
TreeMap-1 & 24 & 24 & 24 & 0.00\% & 0.00\% & 997 & 1190 & 1079 & -8&22\% & 9&33\% \tabularnewline
TreeMap-2 & 148 & 148 & 140 & 5.41\% & 5.41\% & 2918 & 3101 & 2990 & -2&47\%  & 3&58\%\tabularnewline
TreeMap-3 & 1080 & 956 & 806 & 25.37\% & 15.69\% & 21849 & 15112 & 13166 & 39&74\%  & 12&88\%\tabularnewline
BinTree-1 & 84 & 41 & 25 & 70.24\% & 39.02\% & 1476 & 1203 & 1027 & 30&42\%  & 14&63\%\tabularnewline
BinTree-2 & 472 & 238 &  118 & 75.00\% & 50.42\% & 4322 & 4001 & 2974 & 31&19\%  & 25&67\%\tabularnewline
BinTree-3 & 3252 & 1654 &  873 & 73.15\% & 47.22\% & 36581 & 22599 & 16703 & 54&34\%  & 26&09\%\tabularnewline
BinomialHeap-1 & 448 & 38 & 21 & 95.31\% & 44.74\% & 3637 & 2383 & 2054 & 43&52\%  & 13&81\%\tabularnewline
BinomialHeap-2 & 3184  & 218 &  86 & 97.30\% & 60.55\% & 27165 & 8262 & 6235 & 77&05\%  & 24&53\%\tabularnewline
BinomialHeap-3 & 23320  & 1283 & 494 & 97.88\% &  61.50\% & 249224 & 31563 & 22809 & 90&85\%  & 27&74\%\tabularnewline
\hline 
total/average & 33002  & 4918 & 2851 & 91.36\% & 42.03\% & 367056 & 99771 & 78269 & 78&68\%  & 21&55\% \tabularnewline
\hline 
\end{tabular}
\end{table*}

Table \ref{tab:tab1} shows that GreenTrie achieves an  average reuse
ratio that reaches 91.36\%  with respect to symbolic execution without
any constraint reuse, and an average reuse improvement ratio of 42.03\%
compared to Green. In addition, GreenTrie saves an average 21.55\%
of running time with respect to Green and an average 78.68\% of running
time compared to classical symbolic execution without constraint reuse.
The experiment also shows that GreenTrie has better performance in
larger scale program analysis, which has more constraints to be solved
and costs more in symbolic execution time. For small scale of analysis,
GreenTrie may cost a little more time than classical symbolic execution
and Green, but when the scale grows GreenTrie performs better than
Green.

\subsection{Reuse across Runs\label{sub:Reuse-across-Runs}}

This section evaluates the performance of GreenTrie in the scenario
of regression verification. When a changed program is analyzed, the
solution generated by previous runs can be reused in the new run.
We evaluate the performance for three groups of changes: addition,
deletion, and modification. These changes are all small changes and
are generated manually in order to simulate the real situations in
programming. Each group includes 4 version of programs: the first
is the base version, and the others are three changed versions. \textit{Changes
by addition} are generated by adding branches to a program or adding
expressions to program conditions. \textit{Changes by deletion} just
undo changes by addition. \textit{Changes by modification} are generated
by modifying operators or variable assignments. For each group of
changes, we start the evaluation from empty stores, symbolically execute
the base version and the three changed versions of programs one by
one, and evaluate performance figures for each new version of the
program\textit{.}

Tables \ref{tab:tab2-0}, \ref{tab:tab2-1}, \ref{tab:tab2-2} show
the evaluation results for three of the programs we examined in Section
\ref{sub:Reuse-in-Single}. The results show that the average reuse
ratios for three programs are 85.11\%, 82.36\% and 98.93\% with respect
to symbolic execution with no constraint reuse, and an average reuse
improvement of 49.93\%, 73.15\% and 66.71\% with respect to Green.
Thus GreenTrie decreases by more than half the number of evaluated
constraints compared to the Green framework. Considering the average
time saving ratio, we obtain values 77.15\%, 65.47\% and 94.70\% against
symbolic execution without constraint reuse and 31.97\%, 54.76\%,
and 29.81\% as average time improvements against Green. It is also
worth noticing that, in some cases of changes by deletion, both Green
and GreenTrie reuse all the constraints, i.e. \textit{$n_{1}$}=\textit{$n_{2}$}=0.
However, in this situation, GreenTrie also saves running time from
0.54\% to 38.42\%.

\begin{table*}
\protect\caption{\label{tab:tab2-0}Experiment result of reuse across runs (program
BinTree)}

\centering{}%
\begin{tabular}{ccccccccccc}
\hline 
Changes & $n_{0}$ & $n_{1}$ & $n_{2}$ & $R$ & $R'$ & $t_{0}$(ms) & $t_{1}$(ms) & $t_{2}$(ms) & $T$ & $T'$\tabularnewline
\hline 
ADD\#1 & 3438 & 1378  & 656 & 80.92\% & 52.39\% & 39380 & 17828 & 11821 & 69.98\%  & 33.69\%\tabularnewline
ADD\#2 & 8026 & 3425 & 2406 & 70.02\% & 29.75\% & 87903 & 44560 & 36732 & 58.21\% & 17.57\%\tabularnewline
ADD\#3 & 9394  & 1080 & 615 & 93.45\% & 43.06\% & 96238 & 20467 & 15853 & 83.53\% & 22.54\% \tabularnewline
DEL\#1 & 8026 & 2202 & 1222 & 84.77\%  & 95.90\%  & 87903 & 29840 & 23527 & 75.55\% & 21.16\%\tabularnewline
DEL\#2 & 3438 & 1163 & 0 & 100.00\%  & 100.00\% & 39380 & 15083 & 3019 & 83.53\% & 79.98\%\tabularnewline
DEL\#3 & 3252 & 0 & 0 & 100.00\% & 0/0 & 36581 & 2785 & 2770 & 96.57\% & 0.54\%\tabularnewline
MOD\#1 & 3252 & 1682 & 1002 & 69.19\% & 40.43\%  & 40112 & 21997 & 16380 &  59.16\% & 25.54\%\tabularnewline
MOD\#2 & 3252 & 1680 & 632 & 80.57\% & 62.38\% & 39943  & 20510 & 10692 & 73.23\% & 47.87\%\tabularnewline
MOD\#3 & 8296  & 2375 & 970 & 88.31\% & 59.16\% & 97585 & 36794 & 21976 & 77.48\% & 40.27\%\tabularnewline
\hline 
total/average & 50374  & 14985 & 7503 & 85.11\% & 49.93\% & 624682 & 209864 & 142770 & 77.15\% & 31.97\% \tabularnewline
\hline 
\end{tabular}
\end{table*}

\begin{table*}
\protect\caption{\label{tab:tab2-1}Experiment result of reuse across runs (program
Euclid)}

\centering{}%
\begin{tabular}{ccccccccccc}
\hline 
Changes & $n_{0}$ & $n_{1}$ & $n_{2}$ & $R$ & $R'$ & $t_{0}$(ms) & $t_{1}$(ms) & $t_{2}$(ms) & $T$ & $T'$\tabularnewline
\hline 
ADD\#1 & 280 & 260 & 2 & 99.29\%  & 99.23\% & 5057 & 5207 & 1136 & 77.54\% & 78.18\%\tabularnewline
ADD\#2 & 390 & 64 & 11 & 97.18\% & 82.81\% & 6350 & 1704 & 1217 & 80.83\% & 28.58\%\tabularnewline
ADD\#3 & 404 & 325 & 16 & 96.04\% & 95.08\% & 6719 & 6070 & 1432 & 78.69\% & 76.41\%\tabularnewline
DEL\#1 & 390 & 308 & 32 & 91.79\% & 89.61\% & 6350 & 5822 & 1822 &  71.31\% & 68.70\%\tabularnewline
DEL\#2 & 280 & 0 & 0 & 100.00\%  & 0/0 & 5157 & 854 & 781 &  87.70\% & 8.55\%\tabularnewline
DEL\#3 & 278 & 249 & 1 & 96.64\% & 99.60\% & 5105  & 4539 & 728 & 85.88\% & 83.96\%\tabularnewline
MOD\#1 & 260 & 231 & 154 & 40.77\% & 33.33\% & 4610  & 4615 & 3407 &  26.10\% & 26.18\%\tabularnewline
MOD\#2 & 260 & 231 & 174 & 33.08\%  & 24.68\% & 4110  & 4320 & 3568 & 13.19\% & 17.41\%\tabularnewline
MOD\#3 & 150 & 101 & 85 & 43.33\%  & 15.84\% & 2582  & 2007 & 1806 &  30.05\% & 10.01\% \tabularnewline
\hline 
total/average & 2692 & 1769 & 475 & 82.36\% & 73.15\% & 46040 & 35138 & 15897 & 65.47\% & 54.76\% \tabularnewline
\hline 
\end{tabular}
\end{table*}

\begin{table*}
\protect\caption{\label{tab:tab2-2}Experiment result of reuse across runs(program
TCAS)}

\centering{}%
\begin{tabular}{ccccccccccc}
\hline 
Changes & $n_{0}$ & $n_{1}$ & $n_{2}$ & $R$ & $R'$ & $t_{0}$(ms) & $t_{1}$(ms) & $t_{2}$(ms) & $T$ & $T'$\tabularnewline
\hline 
ADD\#1 & 2610 & 86 & 34 & 98.70\% & 60.47\% & 55180 & 4479 & 3164 & 94.27\% & 29.36\%\tabularnewline
ADD\#2 & 2920 & 146 & 92 & 96.85\% & 36.99\%  & 59809 & 4697 & 3758 & 93.72\% & 19.99\%\tabularnewline
ADD\#3 & 6730 & 63 & 20 & 99.70\% & 68.25\% & 125879 & 4013 & 3015 & 97.60\% &  24.87\%\tabularnewline
DEL\#1 & 2920 & 55 & 0 & 100.00\%  & 100.00\% & 59809 & 4930 & 3762 & 93.71\% &  23.69\%\tabularnewline
DEL\#2 & 2610 & 0  & 0  & 100.00\%  & 0/0 & 55180 & 3333 & 2280 & 96.19\% & 31.59\% \tabularnewline
DEL\#3 & 680 & 0  & 0  & 100.00\%  & 0/0 & 12742 & 2709 & 1673 & 96.97\% & 38.24\%\tabularnewline
MOD\#1 & 2024 & 188 & 37 & 98.17\%  & 80.32\% & 29309 & 3756 & 2241 & 92.35\% & 40.34\%\tabularnewline
MOD\#2 & 1494 & 177 & 63 & 95.78\% & 64.41\% & 28464 & 3424 & 2208 & 92.24\% & 35.51\%\tabularnewline
MOD\#3 & 927 & 24 & 0 &  100.00\% & 100.00\% & 19631 & 2324 & 1527 & 92.22\% &  34.29\%\tabularnewline
\hline 
total/average & 22915 & 739 & 246 & 98.93\% & 66.71\% & 446003 & 33665 & 23628 & 94.70\% & 29.81\% \tabularnewline
\hline 
\end{tabular}
\end{table*}

\subsection{Reuse across Programs}

\begin{table*}
\protect\caption{\label{tab:tab3}Experiment result of reuse across programs}

\centering{}%
\begin{tabular}{lllllll}
\hline 
Program & Trityp & Euclid & TCAS & TreeMap & BinTree & BinomialHeap\tabularnewline
\hline 
Trityp & / & 0, 3 & 0, 3 & 0, 4 & 0, 2 & 0, 2\tabularnewline
Euclid & 0, 1 & / & 2, 5 & 0, 0 & 0, 4 & 0, 2\tabularnewline
TCAS & 0, 2 & 2, 2 & / & 0, 0 & 0, 3 & 0, 4\tabularnewline
TreeMap & 0, 0 & 0, 0 & 0, 0 & / & 256, 323 & 0, 0\tabularnewline
BinTree & 0, 0 & 0, 0 & 0, 0 & 256, 470 & / & 0, 1\tabularnewline
BinomialHeap & 2, 5 & 2, 5 & 2, 6 & 0, 3 & 1, 10 & /\tabularnewline
\hline 
\end{tabular}
\end{table*}

Constraint solutions can also be reused across different programs,
especially for the programs which have similar functionality. Our
experiments compare the inter-programs reuse of GreenTrie and Green.
We take six programs in pairs. For every pair, we start with empty
stores, and then symbolically execute one program after the other.
We record the number of reused constraint solutions, which are produced
by the first program and reused in the second, both for GreenTrie
and Green.

The results are shown in Table \ref{tab:tab3}. The first-run programs
are listed in the leftmost column and the second-run programs in the
top row. Each of the cells contains two numbers separated by a comma.
The former is the number of reused constraints when  Green  is used
and the latter is the number of reused constraints when  GreenTrie
is used. As shown in Table \ref{tab:tab3}, when a program pair has
a high level of reuse in Green, GreenTrie has an even higher level
of reuse. And when two programs share almost no constraints in Green,
GreenTrie has a few constraints to reuse.

\section{Related Work}

Our work is closely related to the Green framework, but also has some
relations with other works on constraint solution reuse and constraint
reduction. These are briefly discussed in this section.

\subsection{Reuse of Constraint Solutions }

The idea of improving the speed of constraint solving by reusing previously
solved results is not new. For example, the KLEE \cite{Cadar2008}
symbolic execution tool provides a constraint solving optimization
approach named \textit{counterexample caching}, which stores results
into a cache that maps constraint sets to concrete variable assignments
(or a special No solution flag if the constraint set is unsatisfiable).
For example, \{x + y < 10, x > 5, y \ensuremath{\ge} 0\} maps to \{x
= 6, y = 3\}, and \{i < 10, i = 10\} maps to No. Using these mappings,
KLEE can quickly answer several types of similar queries, involving
subsets and supersets of the constraint sets already cached. The constraint
set \{i < 10, i = 10, j > 12\} is quickly determined to be unsatisfiable
because it has a subset \{i < 10, i = 10\} which is unsatisfiable.
Likewise, \{x + y < 10, x > 5\} is found to be satisfiable and has
a solution \{x = 6, y = 3\} because it is a superset of \{x + y <
10, x > 5, y \ensuremath{\ge} 0\}. The subset and superset queries
in KLEE are a special case of ours: our logical subset and superset
queries fully cover KLEE's subset and superset queries.

\textit{Memoized symbolic execution} \cite{Yang2012} caches the symbolic
execution tree into a trie, which records the constraint solving result
for every branch and reuses them in new runs. When applied to regression
analysis, this allows exploration of portions of the program paths
to be skipped, instead of skipping calls to the solver. GreenTrie
and Green could work together with this approach to provide further
reuse across runs and programs and get better reuse even when the
constraints are not same. 

The work described in \cite{Makhdoom2014} proposes an approach to
eliminate constraint solving for unchanged code by checking constraints
using the test suite of a previous version. While in the process of
exploring states, this approach compares and validates each new path
condition with the solution in the test suite of the base version.
If the comparison succeeds, it just adds that test case to the new
test suite. The work described in \cite{Hossain2014} presents a technique
to identify reusable constraint solutions for regression test cases.
The technique finds variables where input values from the previous
version can be reused to execute the regression test path for the
new version. By comparing definitions and uses of a particular variable
between the old and new versions of the application, this technique
determines whether the same constraints for the variable can be (re)used.
GreenTrie is complementary to these approaches, and is able to provide
better reuse when constraints are not syntactically equivalent.

\subsection{Constraint Reduction}

Reducing the constraint into a short one is a popular optimization
approach of SAT/SMT solvers and symbolic executors \cite{Cadar2008,Sen2005,Cadar2006}.
For example, KLEE \cite{Cadar2008} does some constraint reductions
before solving: (1)\textit{ Expression rewriting}: These are classical
techniques used by optimizing compilers: e.g., simple arithmetic simplifications
(x + 0 $\Rightarrow$ x), strength reduction ($x*2^{n}\Rightarrow x<<n$,
where <\textcompwordmark{}< is the bit shift operator), linear simplification
(2{*}x - x $\Rightarrow$ x). (2) \textit{Constraint set simplification}:
KLEE actively simplifies the constraint set when new equality constraints
are added to the constraint set by substituting the value of variables
into the constraints. For example, if constraint x < 10 is followed
by a constraint x = 5, then the first constraint will be simplified
to true and be eliminated by KLEE. (3)\textit{ Implied value concretization}:
KLEE uses the concrete value of a variable to possibly simplify subsequent
constraints by substituting the variable's concrete value. (4) \textit{Constraint
independence}. KLEE divides constraint sets into disjoint independent
subsets based on the symbolic variables they reference. By explicitly
tracking these subsets, KLEE can frequently eliminate irrelevant constraints
prior to sending a query to the constraint solver. 

The slicing and canonization of Green framework is also able to reduce
the constraints. Constraint slicing is based on constraint independence,
and eliminates irrelevant constraints in an incremental way. Canonization
is able to reduce the constraint by expression rewriting with arithmetic
simplifications. Our approach simplifies the constraint set based
on logic relations, therefore it can reduce constraint into a simpler
form after slicing and canonization by Green.

\subsection{Discussion}

The biggest difference between GreenTrie and other approaches is that
it reuses the constraint solving result based on the implication relationship
among constraints. Green\cite{Visser2012}, memoized symbolic execution
\cite{Yang2012}, the approaches presented in \cite{Makhdoom2014},
\cite{Hossain2014} and \cite{Chen2014} are all based on syntactic
or semantic equivalence of constraint, while KLEE\cite{Cadar2008}
reuses constraints based on simple implication relationships\textemdash subset
and superset. GreenTrie includes the capabilities of these approaches
to support reuse of constraint solutions. The benefits have been demonstrated
in this paper by comparing the degree of constraint reuse by GreenTrie
in comparison with Green. 

We also have shown that GreenTrie saves symbolic execution time with
respect to Green. One reason is that, because of its higher reuse
ratio, it invokes the solver less times than Green. Another reason
is that the logical superset and subset querying from L-Trie is performed
as efficiently or even better than querying from Redis in Green. As
shown in the experiments of Section \ref{sub:Reuse-across-Runs},
when both GreenTrie and Green reuse all the constraints, GreenTrie
is still a little faster than Green.

Unlike Green, which uses Redis to store and query solutions, GreenTrie
saves SCS and UCS as two files on disk and loads them into memory
when symbolic execution is started. GreenTrie uses almost the same
memory as Green for symbolic execution. For example, in the case of
Bintree-3 in Section \ref{sub:Reuse-in-Single}, GreenTrie uses 284Mb
memory, and Green uses 288Mb (including 5M due to the Redis process).
Of course, if needed, it is not difficult to publish interfaces to
GreenTrie as standalone services and make constraint solutions reusable
across different computers. GreenTrie also optimizes the space occupied
by L-Tries: each expression is an object (a sub-constraint is also
an expression composed by smaller expressions), and its occurrences
in different constraints in the trie and the IPOG are all references
to this object. Since the constraints in symbolic execution are always
composed by the same group of expressions/sub-constraints, this optimization
significantly decreases the space occupied by L-Tries. As an example,
in the case of Bintree-3 the total size of SCS and UCS stores is 387
Kb for 873 cached constraints composed with 81 expressions. 

GreenTrie has one limitation compared to the original Green framework:
by now GreenTrie is only able to reuse the SAT solving results, and
cannot reuse the model counting results (that are utilized to calculate
path execution probabilities\cite{Geldenhuys2012}) as Green instead
does.

\section{Conclusion and Future Work}

We introduced a new approach to reuse the constraint solving results
in symbolic execution based on their logical relations. We presented
GreenTrie, an extension to the Green framework, which stores constraints
and solutions into two tries indexed by implication partial order
graphs. GreenTrie is able to carry out logical reduction and logical
subset and superset querying for given constraint, to check if any
solutions in stores can be reused. As our experimental results show,
GreenTrie not only saves considerable symbolic execution time with
respect to the case where constraint evaluations are not reused, but
also achieves better reuse and saves significant time with respect
to Green.

Our future work will extend GreenTrie to support more kinds of constraints
other than linear integer constraints, through adding implication
rules and extending query algorithm, as well as introducing the term
rewriting technique\cite{Braione2013} to simplify the complex constraints.
We also plan to make the summaries in compositional symbolic execution\cite{Godefroid2007,Anand2008}
reusable at a finer granularity, considering that the summary is a
disjunctive constraint that composed by pre and post conditions of
paths of target method. This work is part of our long-term efforts
that aim at supporting incremental and agile verification\cite{Ghezzi2013,ghezzi2014requirement,Bianculli2015}.

\section*{Acknowledgments}

We thank Domenico Bianculli, Srdjan Krstic, Giovanni Denaro, Mauro
Pezzè, Pietro Braione for comments and suggestions in various stages
of this work. This work was supported by European Commission, Program
IDEAS-ERC, Project 227977-SMScom, and the National Natural Science
Foundation of China under Grant No. 61272108, No.61373038, No.91118003.

\bibliographystyle{unsrt}
\bibliography{references}

\end{document}